\documentclass[10pt,conference]{IEEEtran}
\usepackage[utf8]{inputenc}
\usepackage{tikz}
\usetikzlibrary{quantikz}
\usepackage{amsfonts,amssymb}
\usepackage{amsmath,amsthm}
\usepackage[ruled,vlined]{algorithm2e}
\usepackage{color}
\usepackage{subfigure}
\usepackage{soul}
\usepackage{hyperref}

\newtheorem{thm}{Theorem}

\newtheorem{defn}{Definition}

\newtheorem{exam}{Example}

\def \Rm#1{\mbox{\rm #1}}
\def \lsem      {\raise1pt\hbox{\Rm {[\kern-.12em[}}}
\def \rsem      {\raise1pt\hbox{\Rm {]\kern-.12em]}}}
\def \sem#1{\mbox{\lsem$#1$\rsem}}

\usepackage{todonotes}

\title{Equivalence Checking of Parameterised Quantum Circuits}

\author{
\IEEEauthorblockN{Xin~Hong}
\IEEEauthorblockA{\textit{Institute of Software} \\
\textit{Chinese Academy of Sciences}\\
Beijing, China \\
hongxin@ios.ac.cn}
\and
\IEEEauthorblockN{Wei-Jia Huang}
\IEEEauthorblockA{\textit{Foxconn Research } \\
Taipei, Taiwan \\
wei-jia.huang@foxconn.com}
\and
\IEEEauthorblockN{Wei-Chen Chien}
\IEEEauthorblockA{\textit{MediaTek, Inc.} \\
Hsinchu, Taiwan \\
owen.chien@mediatek.com}
\and
\IEEEauthorblockN{Yuan Feng}
\IEEEauthorblockA{\textit{Centre for Quantum Software and Information} \\
\textit{University of Technology Sydney}\\
Sydney, Australia \\
yuan.feng@uts.edu.au}
\and
\IEEEauthorblockN{\color{white}[000000]}
\and
\IEEEauthorblockN{Min-Hsiu Hsieh}
\IEEEauthorblockA{\textit{Foxconn Research } \\
Taipei, Taiwan \\
minhsiuh@gmail.com}
\and
\IEEEauthorblockN{Sanjiang Li}
\IEEEauthorblockA{\textit{Centre for Quantum Software and Information} \\
\textit{University of Technology Sydney}\\
Sydney, Australia \\
sanjiang.li@uts.edu.au}
\and
\IEEEauthorblockN{Mingsheng~Ying}
\IEEEauthorblockA{\textit{Institute of Software} \\
\textit{Chinese Academy of Sciences}\\
Beijing, China \\
yingms@ios.ac.cn}
}

\date{March 2022}

\begin{document}
\maketitle
\begin{abstract}
Parameterised quantum circuits (PQCs) hold great promise for demonstrating quantum advantages in practical applications of quantum computation.  Examples of successful applications include the variational quantum eigensolver, the quantum approximate optimisation algorithm, and quantum machine learning.  However, before executing PQCs on real quantum devices, they undergo compilation and optimisation procedures.  Given the inherent error-proneness of these processes, it becomes crucial to verify the equivalence between the original PQC and its compiled or optimised version.  Unfortunately, most existing quantum circuit verifiers cannot directly handle parameterised quantum circuits;  instead, they require parameter substitution to perform verification.

In this paper, we address the critical challenge of equivalence checking for PQCs.  We propose a novel compact representation for PQCs based on tensor decision diagrams.  Leveraging this representation, we present an algorithm for verifying PQC equivalence without the need for instantiation.  Our approach ensures both effectiveness and efficiency, as confirmed by experimental evaluations.  The decision-diagram representations offer a powerful tool for analysing and verifying parameterised quantum circuits, bridging the gap between theoretical models and practical implementations.

\end{abstract}

\section{Introduction}

In recent years, quantum computing has undergone rapid development due to its potential to exponentially accelerate many classical algorithms. Among these advancements, methods utilising parameterised quantum circuits (PQCs) or variational quantum circuits have emerged as promising candidates for practical applications \cite{mitarai2018quantum,cerezo2021variational}.

PQCs are the quantum counterpart of classical neural networks, which have been used in almost every aspect of our life.  Besides a set of standard quantum gates like single-qubit Pauli gates and two-qubit CNOT or CZ gates, a PQC also employs a set of \emph{parameterised}, typically, single-qubit rotation gates. These parameters are variables that can be adjusted during the computation. This flexibility allows PQCs to be versatile and adaptable for various quantum computing tasks, especially those involving optimisation. The key idea here is to create a family of quantum circuits where the parameters can be tuned to optimise the circuit's performance for a specific task. By adjusting these parameters, a PQC can be tailored to solve particular problems or find optimal solutions more efficiently. This makes them particularly useful for quantum algorithms that involve optimisation tasks, such as the variational quantum eigensolver \cite{Peruzzo_2014}, the quantum approximate optimisation algorithm \cite{Farhi_2014}, and quantum machine learning \cite{benedetti2019Parameterized,mitarai2018quantum}.


Usually, a PQC is designed without a specific target quantum device in mind. This means that, when executing a PQC on a real quantum device, the circuit must be compiled or optimised to meet the connectivity constraints and characteristics of the quantum device. This transformation results in a modified PQC. As these compilation and optimisation processes are susceptible to errors, it becomes crucial to verify the equivalence between the original PQC and the modified version.

For each fixed set of values, we can verify the equivalence of two PQCs by regarding them as standard quantum circuits. For example, we can employ the decision diagram-based methods such as QuIDD \cite{viamontes2003improving}, QMDD \cite{niemann2015qmdds} and the recent tensor decision diagram (TDD) \cite{hong2020tensor}, which all provide canonical representations for quantum circuits. However, it becomes challenging to ensure equivalence across all possible parameter values. At the moment, none of these DDs supports a direct representation of PQCs.




As far as we know, \cite{xu2022quartz} and \cite{peham22equivalence} are the only works devoted to PQC equivalence checking. In \cite{xu2022quartz}, the equivalence of two PQCs is checked by first representing the circuits as symbolic matrices and then calling an SMT solver. In this method, the size of the symbolic matrix and the checking time both grow exponentially with the number of qubits, restricting the scale of the circuits that can be checked to only a few qubits. 
ZX calculus \cite{van2020zx} has been used to check the equivalence of PQCs. This method first transforms PQCs into ZX diagrams and then checks their equivalence by simplifying the diagrams, thus avoiding the use of matrices with exponential dimensions. 
While the rewriting approach may be incomplete, a negative answer does not necessarily imply that the circuits in question are non-equivalent. To mitigate false positives, the method uses multiple non-random instantiations for circuit equivalence checks. This may result in complex circuits for the equivalence checker, additional process steps are then introduced.\cite{peham22equivalence}. 




Recall that decision diagram-based methods have shown great efficiency in the equivalence checking of both classical circuits \cite{molitor2007equivalence} and quantum circuits \cite{burgholzer2020improved,hong2021approximate}. The following question arises naturally: \emph{Can decision diagram-based methods help in PQC equivalence checking}? This paper is devoted to answering this question. We start with a canonical representation of trigonometric polynomials, which are the main symbolic parts in PQCs, and then propose a symbolic TDD by extending TDD \cite{hong2020tensor} with trigonometric polynomial weights. Analogous to TDD, we prove that symbolic TDD provides a compact and canonical representation of PQCs. Experiments on benchmarks with up to 65 qubits, 5135 gates, and 910 parameters demonstrate the efficiency of the symbolic TDD in checking PQC equivalence.


The remainder of this paper proceeds as follows. In section \ref{sec:background}, we provide basic concepts of quantum computing and PQCs. In section \ref{sec:re_vqc}, we discuss the representation of PQCs and propose the symbolic TDD. In section \ref{sec:ec_vqc}, we investigate the equivalence checking of PQCs. The performance of the symbolic TDD-based method is demonstrated by a series of experiments in section \ref{sec:experiments}. Section \ref{sec:conclusion} then concludes the paper.

\section{Background}\label{sec:background}

\subsection{Quantum Computing}

The most fundamental concept in quantum computing is the qubit, the counterpart of the classical bit. While a classical bit can be either 0 or 1, a qubit can be in a superposition of two basis states. Using the Dirac notation, a qubit state can be represented as $\ket{\psi}=a\ket{0}+b\ket{1}$, where $a, b$ are complex numbers with $|a|^2+|b|^2=1$. It can also be represented by the vector $\ket{\psi}=[a,b]^T$. Furthermore, an $n$-qubit state can be represented by a $2^n$-dimensional vector $\ket{\psi}=[\alpha_0,\cdots,\alpha_{2^n-1}]^T$. 

Quantum states evolve through the applications of quantum gates. Each quantum gate corresponds to a unitary matrix. 
A single qubit gate corresponds to a two-dimensional matrix, and an $n$-qubit gate corresponds to a $2^n$-dimensional matrix. An $n$-qubit state $\ket{\psi}$ evolves through the application of an $n$-qubit quantum gate $U$ as $\ket{\psi'}=U\ket{\psi}$.









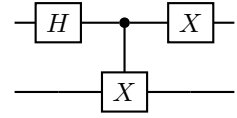
\begin{figure}[h]
\centerline{
\begin{quantikz}[column sep=0.28cm,row sep=0.4cm]
  &\gate{H} &\ctrl{1} &\gate{X}   &\qw\\
  &\qw              &\gate{X}      &\qw             &\qw     \\
\end{quantikz}
}
\caption{An example for quantum circuit.}
\label{fig:cir_qc1}
\end{figure}

A quantum circuit consists of a sequence of quantum gates. Given a quantum state, these quantum gates are applied in a sequential manner. For example, Fig. \ref{fig:cir_qc1} gives a simple quantum circuit consisting of two qubits and three quantum gates. The whole transformation implemented by this circuit is
$$U=\frac{1}{\sqrt{2}}\begin{bmatrix}
	0 & 1 & 0 & -1\\ 
	1 & 0 & -1 & 0\\
	1 & 0 & 1 & 0\\
	0 & 1 & 0 & 1
\end{bmatrix}.$$

\subsection{Parameterised Quantum Circuits (PQCs)}

PQCs are quantum circuits extended with parameterised gates, such as 
\begin{eqnarray*}
    R_z(\theta) &=& \begin{bmatrix}e^{-i\theta/2}&0\\
    0&e^{i\theta/2}
\end{bmatrix},\\
R_y(\theta)&=& \begin{bmatrix}\cos(\frac{\theta}{2})&-\sin(\frac{\theta}{2})\\
\sin(\frac{\theta}{2})&\cos(\frac{\theta}{2})
\end{bmatrix},\\ 
R_x(\theta)&=&\begin{bmatrix}\cos(\frac{\theta}{2})&-i\sin(\frac{\theta}{2})\\
-i\sin(\frac{\theta}{2})&\cos(\frac{\theta}{2})
\end{bmatrix}.
\end{eqnarray*}
Fig. \ref{fig:cir_vqc1} gives a simple PQC.

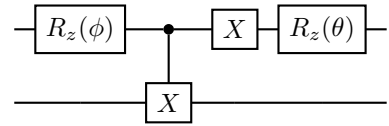
\begin{figure}[h]
\centerline{
\begin{quantikz}[column sep=0.28cm,row sep=0.4cm]
  &\gate{R_z(\phi)} &\ctrl{1} &\gate{X} &\gate{R_z(\theta)}  &\qw\\
  &\qw              &\gate{X} &\qw      &\qw             &\qw     \\
\end{quantikz}
}
\caption{A parameterised quantum circuit.}
\label{fig:cir_vqc1}
\end{figure}

When the values of the parameters are given, an $n$-qubit PQC reduces to an $n$-qubit conventional quantum circuit, thus representing a $2^n\times 2^n$ complex value matrix. More concisely, a PQC can be represented by a $2^n\times 2^n$-symbolic matrix. For example, the circuit shown in Fig. \ref{fig:cir_vqc1} has the following symbolic matrix representation:
$$
\begin{bmatrix}
	0 & 0 & 0 & e^{-i\frac{\theta-\phi}{2}}\\ 
	0 & 0 & e^{-i\frac{\theta-\phi}{2}} & 0\\
	e^{i\frac{\theta-\phi}{2}} & 0 & 0 & 0\\
	0 & e^{i\frac{\theta-\phi}{2}} & 0 & 0
\end{bmatrix}.
$$

\section{The TDD Representation of PQCs}\label{sec:re_vqc}

\subsection{Tensor Decision Diagram}
A predominant method of equivalence checking of conventional quantum circuits is to represent quantum circuits, i.e., unitary matrices, as decision diagrams. 

Tensor Decision Diagram (TDD) \cite{hong2020tensor} is a decision diagram that can be used to represent tensors. The basic operations of tensors, such as contraction and addition, can all be implemented for TDD, thus making it suitable for doing the contraction of tensor networks.  Since quantum circuits are special tensor networks, TDD has also been used in the simulation and equivalence checking of quantum circuits \cite{hong2020tensor,hong2021approximate,hong2023decision}.

A tensor is a multidimensional map ${\phi_{x_1\ldots x_n}: \{0,1\}^n \to \mathbb{C}}$, where $x_1,\cdots,x_n$ are the indices of the tensor and $\phi_{x_1\ldots x_n}(a_1, \ldots, a_n)$ its value for the evaluation $x_1=a_1,\cdots,x_n=a_n$. The classical Boole-Shannon expansion also extends to tensors:
 ${\phi = \overline{x_i} \cdot \phi|_{x_i = 0} + x_i \cdot \phi|_{x_i = 1}}$, where $\overline{x_i}=1-x_i$ for $x_i\in\{0,1\}$ and 
 $\phi|_{x_i = b}$ is the restriction of the tensor $\phi$ to the case when $x_i=b$, $b\in \{0,1\}$.

A TDD is a directed acyclic graph $\mathcal{T}=(V,E,idx,val,w)$, where the node set $V$ consists of non-terminal nodes in $V_N$ and terminal ones in $V_T$ and each non-terminal node $v$ has two child nodes $low(v)$ and $high(v)$; the associated functions are $idx:V_N\rightarrow I$, $val:V_T\rightarrow\mathbb{C}$, and $w:E\rightarrow\mathbb{C}$ assigning,  respectively, node index, terminal node value, and edge weight \cite{hong2020tensor}, where $I$ is the set of indices. In the reduced and normalised form, there is a unique terminal node associated with the value $1$ and a general TDD may experience a series of normalisation and reduction procedures to be transformed into a reduced and normalised TDD. The tensor value of an evaluation $x_1=a_1,\cdots,x_n=a_n$  can be obtained by multiplying the weights of the edges along the path assigned by the evaluation.



\subsection{$S$-valued Tensor Decision Diagram}

As a PQC corresponds to a symbolic matrix, in the following, we extend TDDs with symbolic weights to give a canonical and compact representation of PQCs. More specifically, we assume that these weights are taken values from a ring $S$.

A tensor over a ring $S$, also called an $S$-valued tensor, can be defined as a map $\phi_{x_1,\cdots,x_n}:\{0,1\}^{n} \to S$ where $I=\{x_1,\cdots,x_n\}$ represents a set of indices. Similar to the conventional tensor, an $S$-valued tensor can be written as $\phi = \overline{x_i} \cdot \phi|_{x_i = 0} + x_i \cdot \phi|_{x_i = 1}$ for an index $x_i$ of the tensor. Here $x_i$ is only used to indicate the position of the elements.


Given a ring $(S,+,\cdot)$, symbolic TDDs on $S$ ($S$-TDDs for short) are an extension of TDDs. They can also be regarded as a weighted version of the Algebraic Decision Diagram \cite{bahar1997algebric}.

\begin{defn}[$S$-TDD]
	An \emph{$S$-TDD} $\mathcal{F}$ over a set of indices $I$ and a ring $(S,+,\cdot)$ is a rooted, weighted, and directed acyclic graph
	$\mathcal{F} = (V, E, index, low, high, w)$ defined as follows:
	\begin{itemize}
		\item $V$ is a finite set of nodes which is partitioned into non-terminal nodes $V_N$ and a terminal one $V_T$. Denote by $r_\mathcal{F}$ the unique root node of $\mathcal{F}$;
		\item $index: V_N \rightarrow I$ assigns each non-terminal node an index in $I$;
		\item $low$ and $high$ are mappings in $V_N \rightarrow V$ which assign each non-terminal node with its 0- and 1-successors, respectively;
		\item $E = \{(v, low(v)), (v, high(v)) : v\in V_N\}$ is the set of edges, where $(v, low(v))$ and $(v, high(v)) $ are called the low- and high-edges of $v$, respectively. For simplicity, we also assume the root node $r_\mathcal{F}$ has a unique incoming edge, denoted  $e_r$, which has no source node;
		\item $w: E\rightarrow S$ assigns each edge a weight in $S$. In particular, $w(e_r)$ is called the weight of $\mathcal{F}$, and denoted $w_\mathcal{F}$.
	\end{itemize} 
\end{defn}

Since $S$-TDD is an extension of the original TDD such that all its weights are taken values from $S$, an $S$-TDD represents an $S$-valued tensor. A value of the tensor represented by an $S$-TDD can be obtained by multiplying the weights along a path. More specifically, let $v$ be a node of an $S$-TDD $\mathcal{F}$. Then the $S$-valued tensor $\phi(v)$ represented by it is defined as follows. If $v$ is a terminal node, then $\phi(v):=value(v)$; otherwise, 
$
\phi(v) :=  \overline{x_v} \cdot w_0\cdot \phi(low(v))
+x_v \cdot w_1 \cdot \phi(high(v)),
$ 
where $x_v=index(v)$, and $w_0, w_1 \in S$ are the weights on the low- and high-edges of $v$, respectively. Then, the tensor represented by $\mathcal{F}$ itself is defined to be 
$
\phi(\mathcal{F}) := \phi(r_\mathcal{F}) \cdot w_\mathcal{F},
$
where $r_\mathcal{F}$ and $w_\mathcal{F}$ are the root node and the weight of $\mathcal{F}$, respectively.

It is worth noting that when $S$ takes the Boolean ring, each $S$-TDD is exactly a binary decision diagram (BDD); when $S$ represents the field of complex numbers, each $S$-TDD is a TDD; when $S$ represents the boolean symbolic polynomials, the $S$-TDD is symTDD \cite{hong2023decision}. In this paper, we transform all symbolic terms into trigonometric polynomials and then focus on the ring of trigonometric polynomials.



\subsection{Representation of Weights}
In the following, we give a canonical representation of the weights used for checking the equivalence of PQCs.

In our paper, $e^{i \phi}$ will be decomposed to $\cos(\phi)+i\sin(\phi)$, and $\theta/2$ will be replaced by a new symbol for example $\theta^{\prime}$. Similarly, if $\sin(\theta_1+\theta_2)$, $\cos(\theta_1+\theta_2)$ or other terms are needed, they will all be decomposed to the combination of $\sin(\theta_1), \sin(\theta_2)$ and $\cos(\theta_1), \cos(\theta_2)$. For example, the terms $e^{i\frac{\theta-\phi}{2}}$ can be represented as $\cos(\theta')\cos(\phi')+\sin(\theta')\sin(\phi')+i \sin(\theta')\cos(\phi')-i\cos(\theta')\sin(\phi')$.

As a result, each element in the symbolic matrix representation of a PQC can be replaced by a trigonometric polynomial with the form:
\begin{align}\label{eq:elm_Ctrip}
    f(\theta_0,\cdots,\theta_{n-1})=\sum_{k=0}^{K-1}{c_k\prod_{i=0}^{n-1}{ \sin(\theta_i)^{a_{ki}}\cos(\theta_i)^{b_{ki}} }},
\end{align}
were $K> 0$ is the number of terms in this polynomial, $n$ is the number of variables, $c_k \in \mathbb{C}$, and $a_{ki}, b_{ki} \in \mathbb{N}^+$ for all $k \in \{0,\cdots,K-1\}$ and $i\in \{0,\cdots,n-1\}$.

Observing the identity $\cos(\theta_i)^2+\sin(\theta_i)^2=1$, we further assume that the degree of each $\sin(\theta_i)$ term is either 0 or 1. In fact, if the degree $a_{ki}$ of $\sin(\theta_i)$ is larger than or equal to  $2$, the corresponding term $c\cdot \sin(x_i)^{a_{ki}}\cdot f$ is equivalent to $c\cdot \sin(x_i)^{a_{ki}\! \mod\! 2}\cdot [1-\cos(x_i)^2]^{\lfloor a_{ki}/2 \rfloor} \cdot f$, which can be further normalised. 

In the following, we use $f=g$ and $f=_s g$ to indicate, respectively, the functional equivalence and the syntactic equivalence of two trigonometric polynomials. More specifically, $f=g$ if and only if $f$ and $g$ depend on the same set of variables $\boldsymbol{\theta}=(\theta_0,\cdots,\theta_{n-1})$ and for any $\boldsymbol{\theta} \in \mathbb{R}^n$, $f(\boldsymbol{\theta})=g(\boldsymbol{\theta})$, and $f=_s g$ if and only if they are represented using exactly the same string.

We also give an order for the terms; for example, let $\sin(\theta_0)\prec \cos(\theta_0)\prec \cdots \prec \sin(\theta_{n-1}) \prec \cos(\theta_{n-1})$. If two trigonometric polynomials $f=g$, they must have the same smallest term that they depend on. Suppose the smallest term that $f$ and $g$ depend on is $\sin(\theta_i)$. Then $f$ and $g$ can be rewritten as $f=f_0+\sin(\theta_i)f_1$ and $g=g_0+\sin(\theta_i)g_1$, where $f_0, g_0, f_1, g_1$ are four trigonometric polynomials which do not depend on $\sin(\theta_i)$. Then $f=g$ if and only if $f_0=g_0$ and $f_1=g_1$. On the other hand, suppose the smallest term that $f$ and $g$ depend on is $\cos(\theta_i)$. Then $f$ and $g$ can be rewritten as $f=\cos(\theta)^{b_1}f_{b_1}+\cdots+\cos(\theta)^{b_K}f_{b_K}$ and $g=\cos(\theta)^{d_1}g_{d_1}+\cdots+\cos(\theta)^{d_K}g_{d_K^{\prime}}$, where all $f_{b_i}, g_{d_i}$ are trigonometric polynomials which do not depend on $\cos(\theta_i)$. Then $f=g$ if and only if $K=K^{\prime}$, $b_i=d_i$ and  $f_{b_i} = g_{d_i}$ for all $i \in \{1,\cdots K\}$. Here, we assume that $b_1<\cdots<b_K$ and $d_1<\cdots d_{K^{\prime}}$.

In this way, we can give a canonical representation $\varPhi$ of trigonometric polynomials such that $f=g$ iff $\varPhi(f) =_s \varPhi(g)$ as follows. For any complex number $c\in \mathbb{C}$, $\varPhi(c)=c$. If the smallest term that $f$ depends on is $\sin(\theta_i)$ then $\varPhi(f)=\varPhi(f_0)+\sin(\theta_i)\varPhi(f_1)$, and if the smallest term that $f$ depends on is $\cos(\theta_i)$ then $\varPhi(f)=\cos(\theta_i)^{b_1}\varPhi(f_{b_1})+\cdots+\cos(\theta_i)^{b_K}\varPhi(f_{b_K})$. Then, $\varPhi$ is a canonical representation of trigonometric polynomials.

With this representation, we are able to establish a canonical representation for PQCs. In this paper, we will further represent the trigonometric polynomials as decision diagrams, written as TrDD for short. In this representation, every internal node represents a $\sin(\theta_i)$ or $\cos(\theta_i)$ for some $i$, and then different successors represent different sub-trigonometric polynomials. Fig. \ref{fig:symbolic TDD_pqc} (b) and (c) give the TrDD representations for $(\cos(s_0)-i\cdot \sin(s_0))(\cos(s_1)+i\cdot \sin(s_1))$ and $(\cos(s_0)+i\cdot \sin(s_0))(\cos(s_1)-i\cdot \sin(s_1))$.




\subsection{Normalisation and Reduction}

With the canonical representation of trigonometric polynomials, we pave the way for making an $S$-TDD canonical. In the following, we describe the normalisation and reduction procedures.

Normalisation is one of the necessary procedures to make the representation canonical. To do the normalisation of two trigonometric polynomials, we need to define an order between the terms in these polynomials, that is 
\begin{align*}
   \prod_{i=0}^{n-1}{\sin(\theta_i)^{a_{ki}}\cos(\theta_i)^{b_{ki}}} &\prec \prod_{i=0}^{n-1}{\sin(\theta_i)^{a_{li}}\cos(\theta_i)^{b_{li}}}\quad  \mbox{iff} \\ (a_{k0},b_{k0},\cdots, a_{kn-1},b_{kn-1}) &< (a_{m0},b_{m0},\cdots, a_{mn-1},b_{mn-1}) 
\end{align*}
 in the lexicographical order.

The normalisation of $S$-TDD is implemented by extracting the common factor from the two weights (trigonometric polynomial) on the two outgoing edges of every node and putting the common factor onto the incoming edges of the node. In the following, we denote the normalisation procedure as $loc\_norm(f,g)=(h,f^{\prime},g^{\prime})$, where $f, g$ are the original weights, $h$ is the common factor and $f^{\prime},g^{\prime})$ are the remained weights. In our scheme, $h$ is taken to be the biggest term of form $\prod_{i=0}^{n-1} \sin(\theta_i)^{a_i}\cos(\theta_i)^{b_i}$, such that $h|f$, and $h|g$.



An $S$-TDD $\mathcal{T}$ is called fully normalised if for each internal node $v$ with two outgoing edges weights $f_v$ and $g_v$ and any path leading to it with accumulated weights $h_v$ we have $loc\_norm(h_v\cdot f_v,h_v\cdot g_v)=(h_v,f_v,g_v)$.

We further use the following two rules to reduce the nodes of the representation and make the representation canonical and compact.

\begin{itemize}
    \item RR1: Delete a node $v$ if its 0- and 1-successors are both $w$ and its low- and high-edges have the same weight $f$. Meanwhile, redirect the incoming edge of $v$ to $w$, and multiply the weights of these edges by $f$.
    \item RR2: Merge two nodes if they have the same index, the same 0- and 1-successors, and the same weights on the corresponding edges. 
\end{itemize}

An $S$-TDD is called reduced if it is fully normalised and no reduction rule can be further applied.

The following theorem shows that there is, up to isomorphism and a given index order, a unique reduced $S$-TDD for each $S$-valued tensor. Two $S$-TDDs $\mathcal{F}$ and $\mathcal{G}$ w.r.t. the same index order are \emph{isomorphic} if there is a bijection $\sigma$ between the node sets of $\mathcal{F}$ and $\mathcal{G}$ such that, for each node $v$, $v'=\sigma(v)$ and $v$ have the same index, the same weights on their incoming and outgoing edges, and $\sigma(low(v))=low(v')$, $\sigma(high(v))=high(v')$.

\begin{thm}[canonicity]\label{thm:canonicity}
Let $S$ be the ring of trigonometric polynomials, $I=\{x_1,\ldots,x_{n}\}$ be a set of indices. Suppose $\phi$ is an $S$-valued tensor.  Given any index order $\prec$ of $I$, $\phi$ has a unique reduced $S$-TDD representation w.r.t. $\prec$ up to isomorphism.
\end{thm}
\begin{proof} 
We prove this by using induction on the rank of the tensor. Let $\phi=f$ be a rank 0 tensor. Any reduced $S$-TDD of $\phi$ has only one node labelled with $1$ and has weight $f$. Suppose the statement holds for any rank $k$ tensor. Let $\phi$ be a rank $k+1$ tensor. Suppose $\mathcal{F}$, $\mathcal{G}$ are two reduced $S$-TDDs that represent $\phi$ w.r.t. $\prec$. We show they are isomorphic. Let $q$ be the first index under $\prec$. Without loss of generality, we assume $\phi|_{q=0}\neq \phi|_{q=1}$. That is, $q$ is essential in $\phi$.

Write $r_\mathcal{F}$ and $r_\mathcal{G}$ for the root nodes of $\mathcal{F}$ and $\mathcal{G}$.  Then $q$ is the index of both $r_\mathcal{F}$ and $r_\mathcal{G}$. Let $h_1,f_1,g_1$ and $h_2,f_2,g_2$ be the weights of the incoming and two outgoing edges of, respectively, $r_\mathcal{F}$ and $r_\mathcal{G}$. We write $\mathcal{F}_{0}$ and $\mathcal{F}_1$ for the sub-DDs of $\mathcal{F}$ with root nodes the 0- and 1-successors of $r_\mathcal{F}$ and weights $h_1\cdot f_1$ and $h_1\cdot g_1$, respectively. 
The two sub-DDs $\mathcal{G}_0$ and $\mathcal{G}_1$ of $\mathcal{G}$ are defined in a similar way. 
We assert that $\mathcal{F}_{i}$ and $\mathcal{G}_{i}$ are all reduced. Apparently, no reduction rule can be applied as they are sub-DDs of reduced $S$-TDDs. On the other hand, since $\mathcal{F}$ is fully normalised, and through the definition of fully normalisation, we know that the accumulated weight $h_1\cdot f_1$ and $h_1\cdot g_1$ can not be spread downwards. Thus, $\mathcal{F}_{i}$ and $\mathcal{G}_{i}$ are all fully normalised.

It is easy to see that $\mathcal{F}_0$ and $\mathcal{G}_0$ represent the same tensor $\phi|_{q=0}$ and hence are isomorphic by induction hypothesis. Analogously, $\mathcal{F}_1$ and $\mathcal{G}_1$ represent the same tensor $\phi|_{q=1}$ and are also isomorphic.

Let $\sigma_i$ be the isomorphism between $\mathcal{F}_i$ and $\mathcal{G}_i$ for $i=1,2$. Because $\mathcal{F}$ is reduced, for any node $v$ in both $\mathcal{F}_1$ and $\mathcal{F}_2$, we have $\sigma_1(v)=\sigma_2(v)$. Define $\sigma$ as the extension of $\sigma_1$ and $\sigma_2$ by further mapping $r_\mathcal{F}$ to $r_\mathcal{G}$. We show $\sigma:\mathcal{F}\to \mathcal{G}$ is an isomorphism. To this end, we need only show $h_1= h_2$, $f_1= f_2$, and $g_1= g_2$. Since $\mathcal{F}_i\cong\mathcal{G}_i$, we have $h_1\cdot f_1=h_2\cdot f_2$ and $h_1\cdot g_1= h_2\cdot g_2$. Because $\mathcal{F}$, $\mathcal{G}$ are normalised, we have  $(h_1,f_1,g_1)=loc\_norm(h_1\cdot f_1,h_1\cdot g_1)=loc\_norm(h_2\cdot f_2,h_2\cdot g_2)=(h_2,f_2,g_2)$. This shows $\mathcal{F}\cong  \mathcal{G}$.

\end{proof}

\section{Equivalence Checking of PQCs}\label{sec:ec_vqc}

In this section, we consider the equivalence checking of PQCs. The semantics of a PQC $C$ is given in \cite{xu2022quartz}, which is defined to be a mapping $\sem{C}: \mathbb{R}^m \rightarrow \mathbb{C}^{2^n\times 2^n}$, where $n$ is the number of qubits and $m$ is the number of parameters. Then, the equivalence of two PQCs $C_1$ and $C_2$ are defined to be:
$$\forall \vec{p} \in \mathbb{R}^m, \exists \beta \in \mathbb{R}, s.t., \sem{C_1}(\vec{p})=e^{i\beta}\sem{C_2}(\vec{p}),$$
 which means that for any given values of $\vec{p} \in \mathbb{R}^m$, the two circuits are equivalent up to a global phase \cite{xu2022quartz}.

In \cite{xu2022quartz}, the two circuits are represented by symbolic matrices, and the equivalence is checked using the SMT solver over the theory of nonlinear real arithmetic. In our method, the equivalence of PQCs can be checked by comparing their $S$-TDD representations.

\begin{figure}[h]
\centerline{
\begin{quantikz}[column sep=0.28cm,row sep=0.4cm]
   &\qw &\ctrl{1} &\gate{X} &\gate{R_z(\theta-\phi)}  &\qw\\
   &\qw &\gate{X} &\qw      &\qw             &\qw     \\
\end{quantikz}
}
\caption{A PQC equivalent to the circuit shown in Fig. \ref{fig:cir_vqc1}.}
\label{fig:cir_vqc2}
\end{figure}
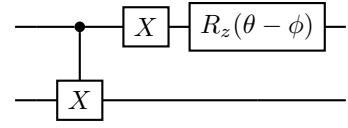

\begin{figure}[h]
    \centering
    \subfigure[]{
   \includegraphics[width=0.125\textwidth]{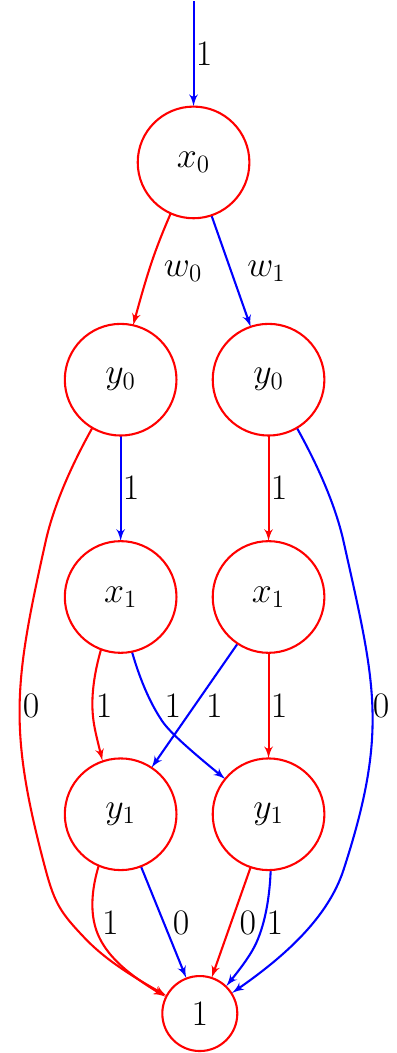}
    }
    \subfigure[]{
    \includegraphics[width=0.1\textwidth]{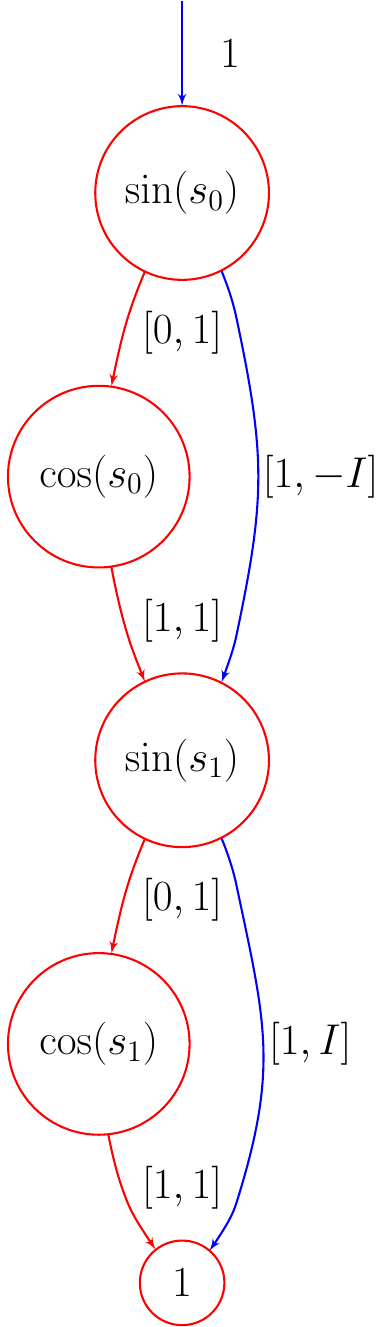}
    }  
    \subfigure[]{
    \includegraphics[width=0.1\textwidth]{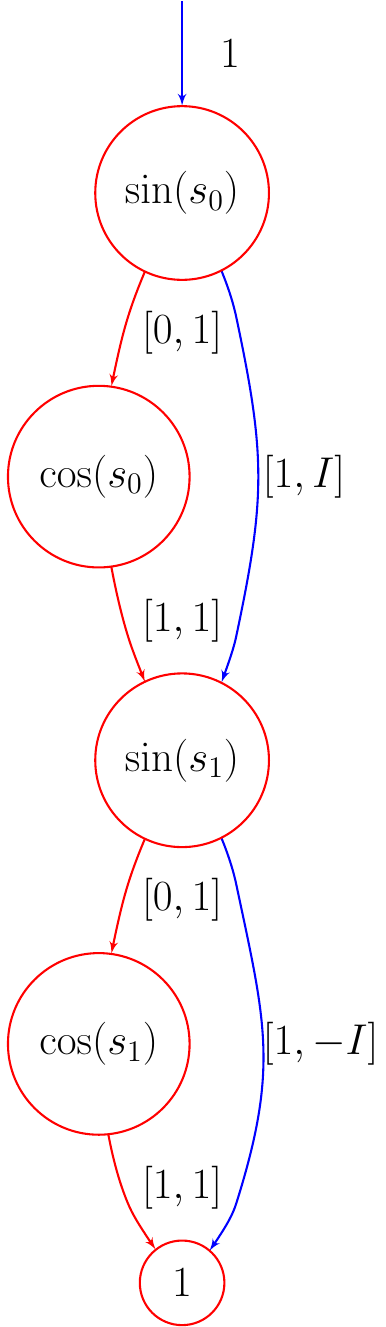}
    }    
    \caption{(a) The symbolic TDD of the PQC shown in Fig. \ref{fig:cir_vqc1} and Fig. \ref{fig:cir_vqc2}; (b) The TrDD representation of $w_0=(\cos(s_0)-i\cdot \sin(s_0))(\cos(s_1)+i\cdot \sin(s_1))$, (c) The TrDD representation of $w_1=(\cos(s_0)+i\cdot \sin(s_0))(\cos(s_1)-i\cdot \sin(s_1))$. The list $[i,c], i\in \mathcal{Z}, c\in \mathcal{C}$ here represent the degree and the coefficient of the term in the node, respectively.}
    \label{fig:symbolic TDD_pqc}
\end{figure}

For example, Fig. \ref{fig:cir_vqc2} gives a PQC equivalent to the circuit shown in Fig. \ref{fig:cir_vqc1}. Both circuits have the $S$-TDD representation as shown in Fig. \ref{fig:symbolic TDD_pqc} (a), where the TrDD representations of $w_0=(\cos(s_0)-i\cdot \sin(s_0))(\cos(s_1)+i\cdot \sin(s_1))$ and $w_1=(\cos(s_0)+i\cdot \sin(s_0))(\cos(s_1)-i\cdot \sin(s_1))$ are shown in Fig. \ref{fig:symbolic TDD_pqc} (b) and Fig. \ref{fig:symbolic TDD_pqc} (c) respectively. This shows that they are equivalent.

However, one more efficient method is to use the technique proposed in \cite{burgholzer2020improved}. The basic idea is to take quantum gates $U_i, V_j$ alternately from the two quantum circuits and then calculate $V_j^{\dagger}U_i$ such that they are as close as an identity matrix. Since identity can be represented very compactly -- normally in $n$ nodes for an $n$ qubits quantum system, this method can reduce the time and space consumption significantly. In our paper, we use the same idea to check the equivalence of two PQCs. We also notice that, in most cases, compilation does not change the order of the parameters that appear in the circuit. In other words, two circuits before and after the compilation can be represented as $U=U_1P_1\cdots U_mP_mU_{m+1}$ and $V=V_1P^{\prime}_1\cdots V_mP_m^{\prime}V_{m+1}$. Here, $U_i, V_j$ are unitary matrices and $P_i$, $P_i^\prime$ are parameterised gates that depend on the same symbol $\theta_i$. Thus, we will first calculate $V_1^{\dagger}U_1$ and then $P_1^{\prime \dagger}V_1^{\dagger}U_1P_1$. After that, we check if the representation is independent with $\theta_1$. If it is so or $\theta_1$ only appears at the weight of the $S$-TDD (serves as the global phase) we continue the process; otherwise, we are certain that the two circuits are not equivalent.


\section{Numerical Experiments}\label{sec:experiments}

In this section, we present the experimental results to demonstrate the efficiency of our method. The experiments were performed on a Standard D8s v3 Azure server equipped with 8 vCPUs and 32 GiB of RAM, and the algorithm was implemented using Python3 and based on the Python implementation of TDD (see \url{https://github.com/Veriqc/TDD}).

We compare our results with the ZX-calculus-based method QCEC \cite{peham22equivalence} for checking the equivalence of PQCs. The benchmark circuits used in the experiments were obtained from the Qiskit circuit library and transpiled into two equivalent circuits. We construct inequivalent circuits by randomly inserting X gates or Z gates into the circuits, as bit flip and phase flip errors are the most common errors for quantum computing. In our experiments, we set the error rate to be 1\%. We check the equivalence of two PQCs using the method proposed above. The benchmark names followed the format ``function name\_qubit number\_repeat times''. The experiment results are presented in Tables~\ref{table:Equivalent check data} and \ref{table: inequivalent check data} below.

\begin{table*}[h!]
\caption{Equivalent check results}
\label{table:Equivalent check data}
\centering
\scalebox{0.9}{
\begin{tabular}{cccccccccccc}
\hline
Benchmark   Name       & $|P|$ & $|Q|$ & $|G1|$ & $|G2|$ & Node\_max & Node\_final & Node\_tdd     & Node\_trdd  & Equivalent  & Time    & QCEC Time \\ \hline
EfficientSU2\_10\_60   & 610 & 10 & 2670  & 2430  & 40  & 31  & 12620  & 3051 & equivalent  & 4.25   & 0.30  \\
RealAmplitudes\_10\_60 & 610 & 10 & 5470  & 3310  & 55  & 31  & 16996  & 3051 & equivalent  & 6.11   & 3.14  \\
TwoLocal\_10\_60       & 610 & 10 & 3610  & 1210  & 40  & 31  & 12483  & 3051 & equivalent  & 3.77   & 0.36  \\
EfficientSU2\_16\_40   & 656 & 16 & 2768  & 2608  & 66  & 49  & 19500  & 3281 & equivalent  & 5.52   & 0.33  \\
RealAmplitudes\_16\_40 & 656 & 16 & 7856  & 5456  & 91  & 49  & 31220  & 3281 & equivalent  & 10.72  & 44.32 \\
TwoLocal\_16\_40       & 656 & 16 & 3856  & 1296  & 66  & 49  & 19054  & 3281 & equivalent  & 5.08   & 0.43  \\
EfficientSU2\_27\_30   & 837 & 27 & 3441  & 3321  & 110 & 82  & 36921  & 4186 & equivalent  & 8.94   & 0.57  \\
RealAmplitudes\_27\_30 & 837 & 27 & 14487 & 11367 & 188 & 82  & 83226  & 4186 & equivalent  & 26.14  & 3.61  \\
TwoLocal\_27\_30       & 837 & 27 & 4887  & 1647  & 110 & 82  & 35906  & 4186 & equivalent  & 8.21   & 7.26  \\
EfficientSU2\_32\_20   & 672 & 32 & 2736  & 2656  & 130 & 97  & 36084  & 3361 & equivalent  & 8.10   & 0.41  \\
RealAmplitudes\_32\_20 & 672 & 32 & 13072 & 10592 & 264 & 97  & 110831 & 3361 & equivalent  & 32.62  & 1.72  \\
TwoLocal\_32\_20       & 672 & 32 & 3872  & 1312  & 130 & 97  & 34726  & 3361 & equivalent  & 7.51   & 0.66  \\
EfficientSU2\_65\_13   & 910 & 65 & 3627  & 3575  & 262 & 196 & 96699  & 4551 & equivalent* & 22.24  & 0.74  \\
RealAmplitudes\_65\_13 & 910 & 65 & 31278 & 27950 & 682 & 196 & 775591 & 4551 & equivalent  & 151.36 & 5.35  \\
TwoLocal\_65\_13       & 910 & 65 & 5135  & 1755  & 262 & 196 & 91852  & 4551 & equivalent* & 20.67  & 2.17     \\ \hline
\end{tabular}
}
\end{table*}

\begin{table*}[h!]
\caption{inequivalent check results}
\label{table: inequivalent check data}
\centering
\scalebox{0.84}{
\begin{tabular}{ccccccccccccccc}
\hline
                       &                &            &            &            &  & \multicolumn{4}{c}{Bit Flip Error}                             &  & \multicolumn{4}{c}{Phase Flip Error}                                  \\ \cline{1-5} \cline{7-10} \cline{12-15}
Benchmark Name         & $|P|$ & $|Q|$ & $|G1|$ & $|G2|$ &  & Node\_tdd & Node\_trdd & Time  & QCEC Time &  & Node\_tdd & Node\_trdd & Time  & QCEC Time \\ \hline
EfficientSU2\_10\_60   & 610 & 10 & 2670  & 2454  &  & 659    & 67  & 0.16  & 0.20  &  & 1261  & 107 & 0.29 & 0.31   \\
RealAmplitudes\_10\_60 & 610 & 10 & 5470  & 3343  &  & 1749   & 102 & 0.58  & 4.99  &  & 2186  & 142 & 0.46 & 5.12   \\
TwoLocal\_10\_60       & 610 & 10 & 3610  & 1222  &  & 1591   & 277 & 0.46  & 0.37  &  & 1369  & 222 & 0.37 & 3.57   \\
EfficientSU2\_16\_40   & 656 & 16 & 2768  & 2634  &  & 827    & 92  & 0.24  & 0.40  &  & 1459  & 112 & 0.32 & 0.34   \\
RealAmplitudes\_16\_40 & 656 & 16 & 7856  & 5510  &  & 1354   & 87  & 0.35  & 17.94 &  & 1083  & 87  & 0.63 & 53.03  \\
TwoLocal\_16\_40       & 656 & 16 & 3856  & 1308  &  & 2174   & 182 & 0.71  & 1.72  &  & 1881  & 152 & 0.37 & 191.21 \\
EfficientSU2\_27\_30   & 837 & 27 & 3441  & 3354  &  & 1593   & 147 & 0.40  & 0.83  &  & 814   & 102 & 0.51 & 0.68   \\
RealAmplitudes\_27\_30 & 837 & 27 & 14487 & 11480 &  & 10392  & 142 & 2.31  & 2.57  &  & 10338 & 142 & 2.52 & 300.53 \\
TwoLocal\_27\_30       & 837 & 27 & 4887  & 1663  &  & 5856   & 387 & 1.49  & 11.01 &  & 3590  & 202 & 0.89 & TO     \\
EfficientSU2\_32\_20   & 672 & 32 & 2736  & 2682  &  & 3116   & 192 & 0.94  & 0.68  &  & 1027  & 127 & 0.31 & 0.62   \\
RealAmplitudes\_32\_20 & 672 & 32 & 13072 & 10697 &  & 20734  & 147 & 4.16  & 3.11  &  & 24540 & 252 & 4.49 & TO     \\
TwoLocal\_32\_20       & 672 & 32 & 3872  & 1325  &  & 7975   & 342 & 1.96  & 2.66  &  & 884   & 127 & 0.20 & TO     \\
EfficientSU2\_65\_13   & 910 & 65 & 3627  & 3610  &  & 2460   & 297 & 0.82  & 2.14  &  & 4068  & 337 & 1.58 & 1.79   \\
RealAmplitudes\_65\_13 & 910 & 65 & 31278 & 28229 &  & 297452 & 332 & 52.37 & 9.75  &  & 1306  & 237 & 0.33 & TO     \\
TwoLocal\_65\_13       & 910 & 65 & 5135  & 1772  &  & 2012   & 292 & 0.75  & 7.42  &  & 2171  & 292 & 0.53 & TO   \\ \hline
\end{tabular}
}
\end{table*}

For the equivalent cases, we record the number of parameters ($|P|$), qubits ($|Q|$) and gates in the two circuits ($|G1|$, $|G2|$) in the second to fifth column. Then, we record the maximum number of nodes and the final number of nodes of the $S$-TDD used to check the equivalence in the sixth and seventh columns. The columns Node\_tdd and Node\_trdd record all the number of nodes used to represent the $S$-TDD and TrDD during the calculation process. The results are then recorded in the tenth column. Here, we use equivalent$^*$ to represent the two circuits equivalent up to a global phase. The last two columns record the time consumption (seconds) of our method and QCEC. The data in Table~\ref{table: inequivalent check data} have similar meanings.

As shown in Table~\ref{table:Equivalent check data}, our proposed method demonstrates excellent performance when applied to PQCs. Specifically, we successfully verified the equivalence of two PQCs for RealAmplitudes\_16\_40, each consisting of 16 qubits and 656 parameters, while circuit 1 contains 7856 gates, and circuit 2 contains 5456 gates. The verification was accomplished within 10.72 seconds, while it took 44.32 seconds for QCEC.

In addition, we observed that our method requires only minimal resources for representing trigonometric polynomials. Hence, our findings indicate that the equivalence of PQCs can be verified with the same level of efficiency as conventional quantum circuits when employing decision diagram-based methods. 
Moreover, for inequivalent circuits, we can obtain the result before finishing the calculating process, thus making it even more efficient. We successfully verified the inequivalence of two PQCs for TwoLocal\_65\_13, each consisting of 65 qubits and 910 parameters, while circuit 1 contains 5153 gates, and circuit 2 contains 1772 gates. The verification process took 0.75 and 0.53 seconds, respectively, while it took 7.42 seconds and even timed out for QCEC. Here, we set time out to be 1200 seconds.


It is worth noting that while our method does not exhibit consistently faster performance than QCEC, the situation will be different if we build our algorithm on the C++ implementation of TDD, which is 1-2 orders of magnitude faster than the Python version of TDD. Furthermore, we observe that our method already outperforms QCEC by a large margin in terms of speed in most of the inequivalence cases.

\section{Conclusion}\label{sec:conclusion}

In our research, we have extensively investigated the intricate challenge of equivalence checking for parameterised quantum circuits (PQCs).  Our innovative approach revolves around a concise and standardized representation known as the $S$-TDD.  This decision graph effectively captures the complex functionality of parameterised quantum circuits, bridging the gap between theoretical models and practical implementations.  We propose a novel algorithm for equivalence checking, which has been experimentally verified to be highly effective.  While it may not surpass the ZX calculus employed by QCEC in verifying circuit equivalence, our approach truly excels when dealing with non-equivalent circuits.  The canonical reduction property of $S$-TDD ensures that there are no false positives, eliminating the need for additional examples to confirm non-equivalence.  Looking ahead, we envision leveraging $S$-TDD to explore more expressive structures within parametric quantum circuits.  Our focus extends to applications such as variational quantum eigensolvers, quantum approximate optimization algorithms, and quantum machine learning.  By advancing the field of equivalence checking, we contribute to enhancing the robustness and reliability of quantum computing systems.

\bibliographystyle{IEEEtran}

\bibliography{referencespqc}

\subsection*{Appendix I: TrDD}\label{sec:appendix3}

TrDD is a graph depicting the representation $\varPhi$. But here, we will also use a weighted version of the representation and the reduction rules as shown in section \ref{sec:re_vqc} will also be used.

\begin{defn}[TrDD]
A TrDD $\mathcal{F}$ over a set of variables $X$ is a rooted, weighted, and directed acyclic graph 
	$\mathcal{F} = (V, E, succ, label, degree, weight)$ defined as follows:
	\begin{itemize}
		\item $V$ is a finite set of nodes, including a set of non-terminal nodes $V_N$ and a single terminal node $V_T$ labelled with an integer $1$;
		\item $label$ assigns each non-terminal node an label $\sin(x)$ or $\cos(x)$ for some $x \in X$;		
		\item $succ$ is a mapping in $V_N \rightarrow V$ which assigns each non-terminal node with its successors;
		\item $E = \{(v, succ(v)) : v\in V_N\} \cup {e_r}$ is the set of edges, and $e_r$ is the edge point to the unique root node of this decision diagram $v_r$;
		\item $degree: E\setminus e_r \rightarrow \mathbb{N}$ assigns a degree to each edge;		
		\item $weight: E\rightarrow \mathbb{C}$ assigns each edge a complex value weight.
	\end{itemize} 
\end{defn}

The trigonometric function $T_{\mathcal{F}}$ represented by a TrDD $\mathcal{F}$ can be defined by $T_{\mathcal{F}}=weight(e_r)*T_{v_r}$, where $T_{v_r}$ is the trigonometric function represented by the root node of $\mathcal{F}$, and
(1) if $v_r$ is the terminal node, then $T_{v_r}=1$;
else, 
(2) $T_{v_r}=\sum_{v_s \in succ(v)}{w_s*label(v_r)^{d_s}*T_{v_s}}$, where $w_s=weight((v_r,v_s))$ and $d_s=degree((v_r,v_s))$.

To make this representation unique, we make sure that for any non-terminal node $v$ and $v_1, v_2 \in succ(v)$, $degree((v,v_1))<degree((v,v_2))$ if $(v,v_1) \prec (v,v_2)$, i.e., the edge $(v,v_1)$ is on the left of $(v,v_2)$ on this graph. And also $weight((v,v_0))=1$, where $(v,v_0)$ is the leftmost edge rooted at $v$. This can be achieved by multiplying the weight on the incoming edge of $v$ by the original $weight((v,v_0))$ and divide the weights $weight((v,v_i))$ by $weight((v,v_0))$ for all $v_i \in succ(v)$. Also, there is an order of all this variables, for example $x_0 \prec x_1\prec \cdots \prec x_{n-1}$, then we can also set an order for all the terms $\sin(x_i)$ and $\cos(x_i)$, say $\sin(x_0)\prec \cos(x_0) \prec \sin(x_1) \prec \cos(x_1)\prec \cdots$. Thus, in our representation all nodes alone a path should respect this order. In addition, we suppose that all nodes and edges that represent the same functionality are merged. For element in $S$, we also ask that $degree((v,v_1))<2$ for all nodes with $label(v)=\sin(x_i)$. Furthermore, we will use similar reduction rules as shown in section \ref{sec:re_vqc} to make the representation as compact as possible.

\begin{exam}
Figure \ref{fig:TrDD_exp} gives an example of TrDD, which represents $1\cdot [1 \cdot \sin(x_0)^0\cdot(1\cdot \cos(x_0)^1 \cdot(1\cdot \cos(x_1)^1)))+i\cdot \sin(x_0)^1\cdot (1\cdot \sin(x_1)^1)]=i\cdot \sin(x_0)\sin(x_1)+\cos(x_0)\cos(x_1)$.
\end{exam}

\begin{figure}[h]
    \centering
    \includegraphics[width=0.1\textwidth]{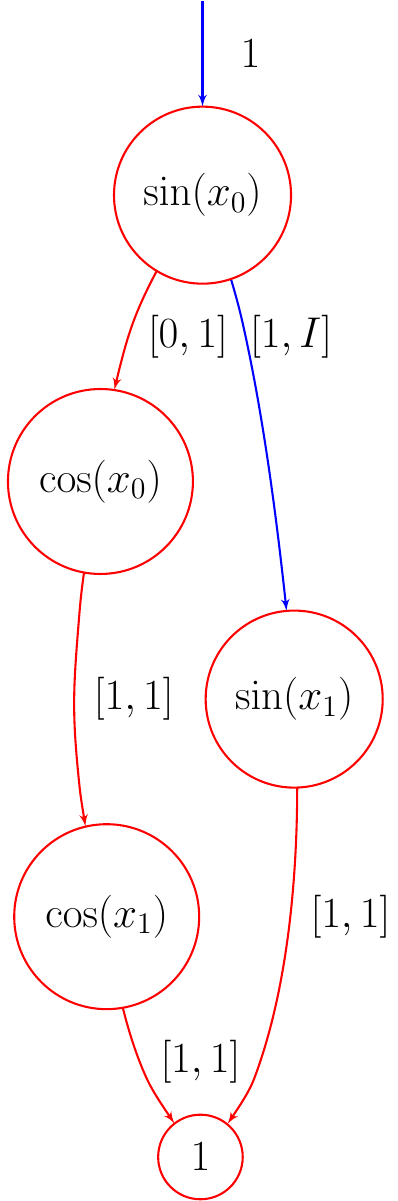}
    \caption{The TrDD representation of $i\cdot \sin(x0)\sin(x1)+\cos(x0)\cos(x1)$.}
    \label{fig:TrDD_exp}
\end{figure}

\begin{thm}\label{thm:trdd_canonicity}
Every element in $S$ can be represented uniquely using TrDD.
\end{thm}
\begin{proof}
The canonicity of TrDD is a direct corollary of the canonicity of the representation $\varPhi$.

First, all elements in the trigonometric polynomial ring $S$ can be represented as a TrDD. 

Then, suppose $f$ is an element in $S$. If it is a constant $c \in \mathbb{C}$, it will be represented as a TrDD with only one node, i.e., the terminal node, and with the incoming edge weight $c$. Then, without loss of generality, suppose the smallest term among $\sin(x_0)\prec \cos(x_0) \prec \sin(x_1) \prec \cos(x_1)\prec \cdots$ that $f$ depends on is $\sin(x_i)$, and suppose that all elements which depend only on terms larger than $\sin(x_i)$ can be represented uniquely. Notice that the biggest degree of $\sin(x_i)$ in $f$ is 1. Then $f$ can be written as $f=\sin(x_i)^0*f_0+\sin(x_i)^1*f_1$, and $f_0$, $f_1$ can be represented uniquely as TrDDs. We also suppose that all nodes in $f_0$ and $f_1$ that represent the same function are merged. Here, $f_0$, $f_1$ are not 0, and we suppose the two weights on the two incoming edges are $w_0$ and $w_1$, respectively.

Then, the TrDD for $f$ can be constructed with a node $v$, labelled with $\sin(x_i)$ and two outgoing edges connecting two the root nodes $v_0$ and $v_1$ for the TrDDs for $f_0$ and $f_1$. Then, $degree((v,v_0))=0$, $degree((v,v_1))=1$, $weight((v,v_0))=1$, $weight((v,v_1))=w_1/w_0$, and the weight of the incoming edge will be set to $w_0$. 

If the smallest term that $f$ depends on is some $\cos(x_i)$, the processes are the same, except that more sub-TrDDs $f_0, \cdots f_k$ should be constructed.
\end{proof}

\end{document}